\def\techreport{}
\documentclass{sig-alternate}

\usepackage{amsthm}
\usepackage{bm}
\usepackage{dsfont}
\usepackage{array}
\usepackage{delarray}

\usepackage{graphicx}
\usepackage{subfigure}
\usepackage{pgfplots}

\usepackage[ruled,vlined,linesnumbered]{algorithm2e}
\usepackage{xspace}
\usepackage{url}
\usepackage{color}
\usepackage{xcolor,colortbl}
\usepackage{booktabs}
\usepackage{multirow}
\usepackage{times}
\usepackage{flushend}

\graphicspath{{./figs/}}
\DeclareGraphicsExtensions{.eps}
\usetikzlibrary{patterns}
\tikzset{
    every mark/.append style={scale=0.7}
}
\pgfplotsset{
    compat=newest,
    tick label style={font=\tiny},
    title style={font=\scriptsize},
    label style={font=\tiny},
    legend style={font=\tiny},
    tick pos=left,
    minor tick num=5,
    /pgfplots/bar cycle list/.style={/pgfplots/cycle list={%
            {blue,fill=blue!30!white,mark=none,postaction={pattern=north east lines}},%
            {red,fill=red!30!white,mark=none,postaction={pattern=north west lines}},%
            {brown!60!black,fill=brown!30!white,mark=none,postaction={pattern=dots}},%
            {black,fill=gray,mark=none,postaction={pattern=grid}},%
    }},
    log x ticks with fixed point/.style={
        xticklabel={
            \pgfkeys{/pgf/fpu=true}
            \pgfmathparse{exp(\tick)}%
            \pgfmathprintnumber[fixed relative, precision=3]{\pgfmathresult}
            \pgfkeys{/pgf/fpu=false}
        }
    },
    log y ticks with fixed point/.style={
        yticklabel={
            \pgfkeys{/pgf/fpu=true}
            \pgfmathparse{exp(\tick)}%
            \pgfmathprintnumber[fixed relative, precision=3]{\pgfmathresult}
            \pgfkeys{/pgf/fpu=false}
        }
    },
}
\newtheorem{thm}{Theorem}[section]

\theoremstyle{definition}

\theoremstyle{remark}

\numberwithin{equation}{section}

\newcommand{\set}[1]{\left\{#1\right\}}

\newcommand{\A}{\bm{A}}
\renewcommand{\P}{\bm{P}}
\newcommand{\p}{\bm{p}}
\newcommand{\e}{\bm{e}}
\newcommand{\q}{\bm{q}}
\newcommand{\x}{\bm{x}}
\newcommand{\y}{\bm{y}}
\newcommand{\s}{\bm{s}}
\newcommand{\f}{\bm{f}}
\newcommand{\pivec}{\bm{\pi}}
\newcommand{\ph}{\hat{\p}}
\newcommand{\pt}{\tilde{\p}}
\newcommand{\RAG}{\text{RAG}}
\newcommand{\sysname}{PowerWalk\xspace}
\newcommand{\vcd}{VERD\xspace}
\renewcommand{\paragraph}[1]{\noindent\textbf{#1}}
\SetKwProg{class}{class}{}{}
\SetKwProg{func}{function}{}{}
\SetKwProg{proc}{procedure}{}{}
\SetKwFunction{decomp}{decomp}
\SetKwFunction{proj}{vc-decomp}
\SetKwFunction{vp}{VertexProgram}
\SetKwFunction{rw}{simulate\_random\_walk}
\SetKwFunction{init}{init}
\SetKwFunction{apply}{apply}
\SetKwFunction{scatter}{scatter}
\SetKwData{walks}{walks}
\SetKwData{counter}{counter}
\SetKwData{prevwalks}{prev\_walks}
\SetKwData{fmap}{f\_map}
\SetKwData{smap}{s\_map}
\SetKwData{prevf}{prev\_f}

\ifdefined\techreport
\makeatletter
\def\@copyrightspace{\relax}
\makeatother
\fi

\begin{document}

\CopyrightYear{2016}
\setcopyright{acmcopyright}
\conferenceinfo{CIKM'16 ,}{October 24-28, 2016, Indianapolis, IN, USA}
\isbn{978-1-4503-4073-1/16/10}\acmPrice{\$15.00}
\doi{http://dx.doi.org/10.1145/2983323.2983713}

\clubpenalty=10000
\widowpenalty = 10000

\title{\sysname: Scalable Personalized PageRank via\\
Random Walks with Vertex-Centric Decomposition}

\numberofauthors{1}

\author{
    \alignauthor
    Qin Liu$^1$, Zhenguo Li$^2$, John C.S. Lui$^1$, Jiefeng Cheng$^2$\\
    \affaddr{$^1$The Chinese University of Hong Kong}\\
    \affaddr{$^2$Huawei Noah's Ark Lab}\\
    \email{$^1$\{qliu, cslui\}@cse.cuhk.edu.hk}\\
    \email{$^2$\{li.zhenguo, cheng.jiefeng\}@huawei.com}
}

\maketitle

\begin{abstract}
    Most methods for Personalized PageRank (PPR) precompute and store all
    \emph{accurate} PPR vectors, and at query time, return the ones of interest
    directly. However, the storage and computation of all accurate PPR vectors
    can be prohibitive for large graphs, especially in
    caching them in memory for real-time online querying.  In this paper, we propose a
    distributed framework that strikes a better balance between \emph{offline
    indexing} and \emph{online querying}.
    The offline indexing attains a fingerprint of the PPR vector of each
    vertex by performing billions of
    ``short'' random walks in parallel across a cluster of machines.  We prove
    that our indexing method has an
    exponential convergence, achieving the same precision with previous methods
    using a much smaller number of random walks.  At query time,
    the new PPR vector is composed by a linear
    combination of related fingerprints, in a highly efficient vertex-centric
    decomposition manner.  Interestingly, \emph{the resulting PPR
    vector is much more accurate than its offline counterpart because it
    actually uses more random walks in its estimation}. More importantly, we
    show that such decomposition for a batch of queries can be very efficiently
    processed using a shared decomposition. Our implementation,
    \emph{\sysname}, takes advantage of advanced distributed graph engines and
    it outperforms the state-of-the-art algorithms by orders of magnitude.
    Particularly, it
    responses to tens of thousands of queries on graphs with billions of edges
    in just a few seconds.
\end{abstract}


\section{Introduction} \label{sec:intro}

Nowadays graph data is ubiquitous, usually in very
large scale.  It is of great interest to analyze these big graphs to gain
insights into their formation and intrinsic structures, which can benefit
services such as information retrieval and recommendations.  Consequently,
large-scale graph analysis becomes popular recently, especially in domains like
social networks, biological networks, and the Internet, where big graphs are
prevalent.  It is found that general-purpose dataflow systems like MapReduce and
Spark are not suitable for graph processing~\cite{Malewicz2010, Gonzalez2012}.  Instead, many graph
computing engines such as Pregel~\cite{Malewicz2010} and
PowerGraph~\cite{Gonzalez2012} are being developed, which
can be orders of magnitude more efficient than general-purpose
dataflow systems~\cite{Malewicz2010,
Gonzalez2012, Gonzalez2014}.  A key feature of most graph engines is the adoption
of a simple yet general \emph{vertex-centric} programming model
that can succinctly express a wide variety of graph
algorithms~\cite{Malewicz2010, Gonzalez2014, Gonzalez2012, Cheng2015}. Here, a graph
algorithm is formulated into a user-defined vertex-program which can be
executed on each vertex in parallel.  Each vertex runs its own instance of the
vertex-program which maintains only local data and affects the data of other
vertices by sending messages to these vertices. To leverage the power
of such a graph engine, it is thus important to formulate the algorithm into a
vertex-centric program.

PageRank is a popular measure of importance of vertices in a
graph~\cite{Page1999}. This is particularly useful in large applications where
data and information are crowded and noisy; a measure of their importance
allows to process them by importance~\cite{Page1999}. One prominent application is
the web search engine of Google where the pages returned, in response to any
query, are ordered by their importance measured by PageRank. The intuition
behind PageRank is that a vertex is important if it is linked to by many
important vertices. This ``global'' view of vertex importance does not reflect
individual preferences, which is inadequate in modern search engines, online
social networks, and e-commerce, where customizable services are in urgent
need. Consequently, ``personalized'' PageRank (PPR)~\cite{Page1999} attracts
considerable attention lately that allows to assign more importance to a
particular vertex in order to provide ``personalized views'' of a graph.  PPR
has been widely used in various tasks across different domains: personalization
of web search~\cite{Page1999}, recommendation in online social
networks~\cite{Gupta2013} and mobile-app marketplaces~\cite{he2015mining},
graph partitioning~\cite{Andersena2006}, and other
applications~\cite{Tong2006}. While the enabling of ``customization'' in
PageRank greatly expands its utility, it also brings a huge challenge on its
computation -- the customization to each vertex leads to a workload $N$ times
of the original PageRank, where $N$ is the number of vertices in the graph, and
the customization on any combination of vertices exponentially increases the complexity.

Despite significant progress, existing approaches to PPR computation are still
restricted in large-scale applications.  While the power iteration
method is used in most graph engines to calculate the global
PageRank~\cite{Malewicz2010, Gonzalez2014, Gonzalez2012, Cheng2015}, it is
impractical for PPR computation, which would take $O(N(N+M))$ time for all PPR
vectors ($N$ and $M$ are the numbers of vertices and edges
respectively)~\cite{Berkhin2005}.  Furthermore, most algorithms for PPR
computation are designed for single-machine in-memory systems which cannot deal
with large-scale problems~\cite{Shin2015, Maehara2014b, Lofgren2014}.
The Hub Decomposition algorithm~\cite{Jeh2003} allows personalization only for
a small set of vertices (hubs).  To achieve full personalization (the
computation of all $N$ PPR vectors), it requires $O(N^2)$ space.
Fogaras et al.~\cite{Fogaras2005} proposed a scalable Monte-Carlo solution to
full personalization.  This approach first simulates a large number of short
random walks from each vertex, and store them in a database.  Then it uses
these random walks to approximate PPR vectors for online queries.  As we will
show in our evaluation (Section~\ref{sec:runtime}), to achieve high accuracy,
the simulation of random walks is very time-consuming.  Also, the precomputed
database is usually too large to fit in the main memory for large graphs which
makes this method undesirable in practice.  Bahmani et al.~\cite{Bahmani2010}
suggested to concatenate the short random walks in the precomputed database to
answer online queries.  This extension can reduce the size of database,
but it still incurs random access to the graph data and precomputed
database.  If implemented on a distributed graph engine, it would require thousands
of iterations to concatenate short random walks together to form a long walk which
makes it inefficient.

One can see that, given the great effort in scaling up PPR, it remains open for
a practically scalable solution, in that: (1) it enables full personalization
for all vertices in the graph; (2) it is scalable for very large problems; (3)
it supports online query efficiently; and (4) it executes a large number of
queries in reasonable time, which is critical to modern online services where
numerous requests are fed to a system simultaneously~\cite{Attenberg2011}.

In this paper, we present a novel framework, \emph{\sysname}, for scalable PPR
computation.  We implement \sysname on distributed graph engines, in order to
harness the power of a computing cluster. Furthermore, to enable fast online
services, we separate the computation into \emph{offline preprocessing} and
\emph{online query}, in a flexible way that \emph{the computation can be
shifted to the offline stage as much as the memory budget allows}.
The offline stage uses a variant of the Monte-Carlo methods~\cite{Fogaras2005,
Avrachenkov2007} to compute an approximation for each PPR vector.  We can tune
the precision of PPR vectors according to the memory budget, so that all
approximate PPR vectors can be cached in distributed memory.  At query time,
the PPR vectors of the querying vertices are computed efficiently based on the
precomputed PPR approximate vectors by utilizing the Decomposition
theorem in~\cite{Jeh2003}.
Our main contributions are:
\begin{itemize}
    \item We propose a Monte-Carlo Full-Path (MCFP) algorithm for offline PPR
        computation, and prove that it converges exponentially fast. In
        practice, it requires only a fraction of random walks to achieve the
        same precision compared with the previous Monte-Carlo
        solution~\cite{Fogaras2005}.
    \item We propose a Vertex-Centric Decomposition (\vcd) algorithm which can
        provide results to online PPR queries in real-time and more accurate than their
        offline counterparts. More importantly, it is able to execute a large number
        of queries very efficiently with a shared decomposition.
    \item Combining the MCFP and \vcd algorithms, we propose an
        efficient distributed framework, \sysname, for computing full
        personalization of PageRank. \sysname can adaptively trade off between
        offline preprocessing and online query, according to the memory budget.
    \item We evaluate \sysname on billion-scale real-world graphs and validate
        its state-of-the-art performance. On the Twitter graph with 1.5 billion
        edges, it responses to 10,000 queries in 3.64~seconds and responses to
        100,000 queries in 17.96~seconds.
\end{itemize}

The rest of the paper is organized as follows.  We elaborate our two algorithms
for PPR computation in Section~\ref{sec:alg}.  Section~\ref{sec:implementation}
presents the implementation of \sysname.  We extensively evaluate \sysname in
Section~\ref{sec:eval}.  Section~\ref{sec:related} reviews related work on PPR
computation.  Finally, Section~\ref{sec:conclusion} concludes the paper.

\section{Algorithm} \label{sec:alg}

In this section, we propose two algorithms for computing \emph{Personalized
PageRank (PPR)} for web-scale problems.  The first algorithm is called the
\emph{Monte-Carlo Full-Path} (MCFP) Algorithm, which is based on the Monte-Carlo
simulation framework and can be used to compute the fully PPR (or all
PPR vectors of all vertices in the graph) in an offline manner.  Our second
algorithm is called the \emph{Vertex-Centric Decomposition} (\vcd)
Algorithm, which is capable of computing the PPR vector for any vertex in an
online manner. As we will show in
Section~\ref{sec:implementation}, the proper coupling of these two algorithms can lead to
a practically scalable solution to the PPR computation which
can response to a large number of queries very efficiently.  Before we proceed
to our algorithms, let us first provide some preliminaries.

\subsection{Preliminaries} \label{sec:pre}

Let $G = (V, E)$ be a directed graph, where $V$ is the set of vertices and $E$
is the set of edges.  An edge $(u, v) \in E$ is considered to be directed from
$u$ to $v$.  We also call $(u, v)$ as an out-edge of $u$ and an in-edge of $v$,
and call $v$ as an out-neighbor of $u$ and $u$ as an in-neighbor of $v$.  Let
$N = |V|$ and $M = |E|$ denote the numbers of vertices and edges respectively,
and $O(v)$ denote the set of out-neighbors of $v$.

Personalized PageRank (PPR) measures the ``importance'' of all vertices from a
perspective of a particular vertex.  Let $\p_u$ be the stochastic row vector
that represents the PPR vector of $u$ (i.e.,  $\p_u$ is a non-negative row
vector with entries summed up to $1$), and $\p_u(v)$ represents the
Personalized PageRank of vertex $v$ from the perspective of the source vertex
$u$.  If $\p_u(v) > \p_u(w)$, this means that, from the perspective of vertex
$u$, vertex $v$ is more important than vertex $w$.  Finally, $\p_u$ can be
defined as the solution of the following equation~\cite{Fogaras2005}:
\begin{equation} \label{eq:ppr}
    \p_u = (1-c) \p_u \A + c \e_u,
\end{equation}
where $\A$ is a row-stochastic matrix\footnote{I.e., $\A$ is non-negative and
the entries in each row sum up to $1$.} and $\A_{i,j}$ equals to $1/|O(i)|$ if
$(i,j) \in E$, $\e_u$ is a unit row vector with the $u$-th entry $\e_u(u)$
being one and zero elsewhere, and $c \in (0, 1)$ is a given \emph{teleport
probability}.
Typically, $c$ is set to $0.15$, and empirical studies show that small
changes in $c$ have little effect in practice~\cite{Page1999}.
If $i$ is a dangling vertex (i.e., one without any out-edge), to
make sure $\A$ is row-stochastic, one typical adjustment is to add an
artificial edge from $i$ to $u$~\cite{Berkhin2005}.  This is equivalent to
setting row $i$ of $\A$ as $\e_u$, and as a result, any visit to vertex $i$
will immediately follow by a visit to vertex $u$.  This reflects the intrinsic
of $\p_u$, that is, it encodes the importance of all vertices from the
perspective of vertex $u$.

The above equilibrium equation has the following random walk
interpretation: a walk starts from vertex $u$; in each jump, it will move to
one random out-neighbor with probability $1-c$, or jump back (or teleport) to
vertex $u$ with probability $c$.  Particularly, if the current vertex is a
dangling vertex, then the next jump will surely go to vertex $u$.  In summary,
the PPR vector $\p_u$ is exactly the stationary distribution of such a random
walk model, and $\p_u(v)$ is the probability of the random walk visiting $v$ at
its equilibrium.

\subsection{The MCFP Algorithm} \label{sec:mc}



The fact that $\p_u(v)$ is the probability of a random walk visiting $v$ at its
equilibrium immediately suggests that one can approximate $\p_u(v)$ by $\x_n(v)
/ n$ for a large $n$, where $\x_n(v)$ denotes the number of visits to $v$ in
$n$ steps.
However, simulating a long random walk (or a random walk with many jumps) on a
large graphs is inefficient for the following reasons.  As a large graph
usually cannot be held in the main memory, it must be partitioned across
multiple machines or to be stored on the disk.  At each step, the random walk
may jump to a vertex at another machine or access data on disk, incurring an expensive
network or disk I/O and latency.  Moreover, a long walk takes many iterations
to reach the steady state solution, which can become a bottleneck for systems like MapReduce.
Finally, there is an inherent restriction of implementing a long random walk on
any graph computation engines: the dependency between adjacent moves prevents
executing the walks in parallel.

\begin{figure}[htbp]
    \centering
    \includegraphics[width=0.48\textwidth]{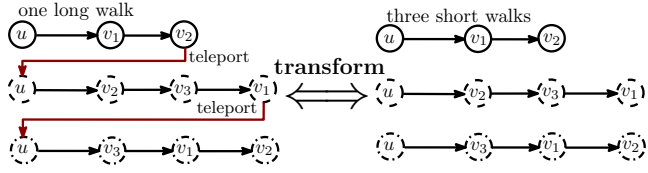}
    \caption{Breaking a long walk into three short ones}
    \label{fig:short-long}
\end{figure}

To overcome this inefficiency, we propose to transform the long random walk
into \emph{many} short walks by breaking it at every teleport move.  This is
illustrated in Figure~\ref{fig:short-long}.  One crucial benefit of doing so is
that these short random walks can be \emph{executed in parallel}.  In the next
section, we will show how to simulate billions of short random walks on
a distributed cluster.  The above idea translates to our Monte-Carlo Full-Path
(MCFP) algorithm for PPR computation as stated in Algorithm~\ref{alg:mc-fp},
and we show in Theorem~\ref{thm:fp-bound} that it converges exponentially fast.

\begin{algorithm}[ht] \label{alg:mc-fp}
    \caption{Monte-Carlo Full-Path (MCFP)}
    \KwIn{vertex $u$, the number of random walks $R$}
    \KwOut{an approximation of $\p_u$}
    Simulate $R$ random walks starting from $u$\;
    At each step, each of the $R$ random walks terminates with probability $c$
    and takes a further step according to the matrix $\A$ with probability
    $1-c$\;
    Approximate $\p_u(v)$ by the fraction of moves resident at
    $v$, i.e., $\p_u(v) \sim \x_n(v) / n$, where $\x_n(v)$ is the number of
    visits by all $R$ random walks to vertex $v$, and $n$ is the total moves in
    all $R$ random walks\;
\end{algorithm}

\begin{thm} \label{thm:fp-bound}
    \label{THM:FP-BOUND} 
    Let $\ph_u$ denote the estimator of $\p_u$ obtained from
    Algorithm~\ref{alg:mc-fp}.  The probability of over-estimating $\p_u$ can
    be bounded as follows:
    \[
        \Pr[\ph_u(v) - \p_u(v) \geq \gamma] \leq \frac{1}{\sqrt{c}}
        \left(1+\frac{\gamma c}{10}\right) \exp\left(\frac{-\gamma^2
        R}{20}\right).
    \]
    The same bound also holds for the probability of under-estimation.
\end{thm}

\ifdefined\techreport
\begin{proof}
    Please refer to Appendix~\ref{sec:proof-1}.
\end{proof}
\else
\begin{proof}
    For detail proof, please refer to our technical
    report~\cite{Liu2016PowerWalk}.
\end{proof}
\fi

\paragraph{Prior Art.}
Let us compare with one current state-of-the-art random walk method in
computing PPR~\cite{Fogaras2005}.  Consider a random walk starting from vertex
$u$.  At each step, the random walk terminates with probability $c$,  or it
jumps to other states according to the matrix $\A$ with probability $1-c$.  It
has been proved in~\cite{Jeh2003} that the last visited vertex of such a random
walk has a distribution $\p_u$ which suggests another Monte-Carlo algorithm in
approximating $\p_u$.  Note that the distribution $\p_u$ is achieved by
simulating the random walk for many times, and each time, the random walk may
end in a different vertex.  So one needs to normalize all such ending vertex
occurrences to have the probability distribution $\p_u$.  This state-of-the-art
algorithm was proposed in~\cite{Fogaras2005} and it is illustrated in
Algorithm~\ref{alg:mc-ep}.

\begin{algorithm}[ht] \label{alg:mc-ep}
    \caption{Monte-Carlo End-Point (MCEP) in~\protect\cite{Fogaras2005}}
    \KwIn{vertex $u$, the number of random walks $R$}
    \KwOut{an approximation of $\p_u$}
    Simulate $R$ random walks starting from $u$\;
    At each step, each of the $R$ random walks terminates with probability $c$
    and takes a further step according to the matrix $\A$ with probability
    $1-c$\;
    Approximate $\p_u(v)$ by the fraction of $R$ random walks that terminate
    at vertex $v$, i.e., $\p_u(v) \sim \y(v) / R$. Here, $\y(v)$
    denotes the number of the random walks terminate at $v$\;
\end{algorithm}


Note that while Algorithm~\ref{alg:mc-ep} takes only the ending vertex of each
random walk into account, our proposed MCFP algorithm leverages the \emph{full
trajectory} of each random walk.  Our motivation is that the intermediate
vertices on each trajectory contain significant information regarding the
distribution of the random walk and should be included in approximating.
Our theoretical analysis stated in Theorem~\ref{thm:fp-bound} guarantees not
only the validity of our algorithm but also the exponential convergence rate.
Our evaluation in Section~\ref{sec:accuracy} confirms that this does lead to a
much better approximation. In other words, our MCFP is far more efficient than the
state-of-the-art under the same precision.   For example, to
achieve the same precision of our MCFP algorithm with 1,000 random
walks, Algorithm~\ref{alg:mc-ep} needs to simulate 6,700 random walks.



\subsection{The \vcd Algorithm} \label{sec:vc-decomp}

With the MCFP algorithm, we can approximate the fully PPR, i.e., we can compute
an approximate vector $\ph_u$ of $\p_u$ for all $u \in V$, which can be carried
out offline.  Depending on our memory budget, we can adjust the precision of
our approximations by varying $R$ used in the MCFP algorithm, so that all
approximations of PPR vectors fit in the distributed memory.  For web-scale
problems and with limited memory budget, this can lead to low-resolution
approximations.  In this section, we propose a method to recover in real-time a
high-resolution PPR vector from low-resolution approximations.  Our method is
inspired by the decomposition theorem~\cite[Theorem~3]{Jeh2003}.

\begin{thm}[Decomposition]
    Suppose $\p_u$ is the PPR vector of vertex $u$. We have
    \begin{equation} \label{eq:decomposition}
        \p_u = c \cdot \e_u + \frac{1-c}{|O(u)|} \sum_{v \in O(u)} \p_v.
    \end{equation}
\end{thm}

This theorem basically says that the PPR vector of a vertex can be recovered
from the PPR vectors of its out-neighbors.  Another more important observation
is that such an approximation is more accurate than those of its out-neighbors,
because it actually uses $R \cdot |O(u)|$ random walks when $\p_v$, $v \in
O(u)$, uses $R$ random walks, and the error of $\p_u$ is reduced by a factor of
$1-c$.  Furthermore, the decomposition theorem can be applied recursively to
achieve higher precision.  This suggests Algorithm~\ref{alg:decomp} in
approximating the PPR vector $\p_u$ of vertex $u$ with $T$ iterations.  When $T=0$
(Lines~\ref{line:t0-1}~\&~\ref{line:t0-2}), the algorithm simply returns the
precomputed approximation obtained from our MCFP algorithm.

\begin{algorithm}[ht]
    \caption{Recursive Decomposition} \label{alg:decomp}
    \func{\decomp{$u$, $T$}}{
        \KwIn{a vertex $u$, the number of iterations $T$}
        \KwOut{an approximation of $\p_u$}
        \If{$T=0$}{ \label{line:t0-1}
            \Return the precomputed $\ph_u$ \label{line:t0-2}
        }
        \Return $c \cdot \e_u + \frac{1-c}{|O(u)|} \sum_{v \in O(u)} \decomp(v, T-1)$\;
    }
\end{algorithm}

However, implementing such a recursive algorithm on distributed graph engines
appears to be very challenging, because existing graph engines only provide
vertex-centric programming interface which does not support recursion.  On the
other hand, there is a huge redundancy in Algorithm~\ref{alg:decomp}
because the out-neighbors of different vertices are likely to
overlap. To illustrate our idea,
consider the graph in Figure~\ref{fig:toy-graph}.  To compute $\p_{v_1}$ with
two iterations of recursion by calling \decomp{$v_1$, $2$}, \decomp{$v_5$,
$0$}, \decomp{$v_6$, $0$}, and \decomp{$v_7$, $0$} will be called twice.

\begin{figure}[htbp]
    \centering
    \includegraphics[width=0.3\textwidth]{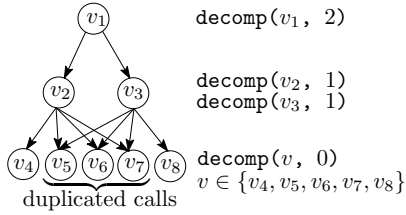}
    \caption{A simple graph to illustrate the \vcd algorithm}
    \label{fig:toy-graph}
\end{figure}

To remove such redundancy in $T$ levels of recursion, we suggest to
\emph{unfold} the $T$ levels of recursion by applying the decomposition in
Equation~\eqref{eq:decomposition} iteratively $T$ times:
\begin{equation} \label{eq:projection}
    \p_u = \s_u^{(T)} + \sum_{v \in V} \f_u^{(T)}(v) \p_v,
\end{equation}
where $\s_u^{(T)}$ and $\f_u^{(T)}$ are two row vectors given by
\begin{align*}
    \s_u^{(T)} & = \s_u^{(T-1)} + \sum_{w \in V} c \cdot \f_u^{(T-1)}(w)\e_w \\
    \f_u^{(T)} & = \sum_{w \in V} \frac{1-c}{|O(w)|} \sum_{v \in O(w)}
    \f_u^{(T-1)}(w) \e_v.
\end{align*}
where
\[
    \s_u^{(0)} = \vec{0}, \f_u^{(0)} = \e_u.
\]

The benefits of the unfolded decomposition in Equation~\eqref{eq:projection} are
several-fold.  First, besides the removal of redundancy, it enables
vertex-centric programming, as we show in Algorithm~\ref{alg:proj}.  Consider
Lines~\ref{line:v-begin}-\ref{line:v-end}, the updates on the two vectors are
carried entirely in a vertex-centric manner, which makes it easy to implement
on existing graph engines as shown in Section~\ref{sec:impl_query}. For this
reason, we call Algorithm~\ref{alg:proj} the \emph{\underline{Ver}tex-Centric
\underline{D}ecomposition} (\vcd) algorithm. Second, since the \vcd algorithm
is iterative, it will not overload the stack like the recursive decomposition
algorithm (Algorithm~\ref{alg:decomp}) when $T$ is large. Finally, as we show
in Section~\ref{sec:impl_query}, it allows to simultaneously compute multiple
PPR vectors very efficiently by sharing network transfer, which is important to
online services.
\ifdefined\techreport
Theorem~\ref{thm:vd} below shows that Algorithm~\ref{alg:proj} is equivalent to
Algorithm~\ref{alg:decomp}.

\begin{thm} \label{thm:vd}
    \label{THM:VD} 
    For $u \in V$ and $T \geq 0$,
    \[
        \decomp(u, T) = \proj(u, T).
    \]
\end{thm}

\begin{proof}
    Please refer to Appendix~\ref{sec:proof-2}.
\end{proof}
\else
In our technical report~\cite{Liu2016PowerWalk}, we show that
Algorithm~\ref{alg:proj} is equivalent to Algorithm~\ref{alg:decomp}.
\fi


\begin{algorithm}[ht]
    \caption{Vertex-centric Decomposition (VERD)} \label{alg:proj}
    \func{\proj{$u$, $T$}}{
        \KwIn{a vertex $u$, the number of iterations $T$}
        \KwOut{an approximation of $\p_u$}
        $\s_u^{(0)} \leftarrow \vec{0}, \f_u^{(0)} \leftarrow \e_u$\;
        \For{$i=1$ \KwTo $T$}{
            $\f_u^{(i)} = \vec{0}$\;
            \ForEach{$w \in V$}{
                \If{$\f_u^{(i-1)}(w) > 0$}{
                    \label{line:v-begin}
                    \ForEach{$v \in O(w)$}{
                        $\f_u^{(i)}(v) \leftarrow \f_u^{(i)}(v) + \frac{1-c}{|O(w)|} \f_u^{(i-1)}(w)$\;
                    }
                    $\s_u^{(i)}(w) \leftarrow \s_u^{(i-1)}(w) + c \cdot \f_u^{(i-1)}(w)$\;
                    \label{line:v-end}
                }
            }
        }
        \Return $\s_u^{(T)} + \sum_{v \in V} \f_u^{(T)}(v) \ph_v$\;
    }
\end{algorithm}

To illustrate how the \vcd algorithm removes redundant calls,
consider again the graph in Figure~\ref{fig:toy-graph}.  Calling \proj{$v_1$,
$2$}, we have $\s_{v_1}^{(2)}$ equals to
\[
    \left(c, \frac{c(1-c)}{2}, \frac{c(1-c)}{2}, 0, 0, 0, 0, 0\right)
\]
and $\f_{v_1}^{(2)}$ equals to
\[
     \left(0, 0, 0, \frac{(1-c)^2}{8}, \frac{(1-c)^2}{4}, \frac{(1-c)^2}{4},
     \frac{(1-c)^2}{4}, \frac{(1-c)^2}{8}\right).
\]
Then the \vcd algorithm simply returns a new approximation of
$\p_{v_1}$ by computing
\[ s_{v_1}^{(2)} + \sum_{i=4}^8 f_{v_1}^{(2)}(i) \ph_{v_i}. \]
In this process, each involved precomputed approximation is loaded only once. In contrast, in the
recursive decomposition algorithm, we have to load $\ph_{v_5}$, $\ph_{v_6}$,
and $\ph_{v_7}$ twice.


The sizes of $\s_u^{(T)}$ and $\f_u^{(T)}$ grow with the number of iterations
$T$, and in the worst case, they can be up to the size of $\p_u$, which is the
number of vertices reachable from $u$.  In practice, when $T$ is small,
$\s_u^{(T)}$ and $\f_u^{(T)}$ are usually highly sparse which can be stored in
hash tables or sorted vectors.  Also, we can discard values below a certain
threshold $\epsilon$ in $\s_u^{(T)}$ and $\f_u^{(T)}$ to keep them sparse when
$T$ becomes larger.  On distributed graph engines, network communication is
mainly caused by the synchronization at each iteration.  As we will show in
Section~\ref{sec:accuracy}, by utilizing the precomputed approximations, our
\vcd algorithm needs just a few iterations to achieve reasonable
accuracy, which makes it quite efficient for online query.

\paragraph{Prior Art.}
In Algorithm~\ref{alg:proj}, when $T \rightarrow \infty$, we have $\f_u^{(T)}
\rightarrow \bm{0}$.  So when $T$ is large enough, Algorithm~\ref{alg:proj} can
be used to compute PPR vectors without any precomputed approximations by simply
returning $\s_u^{(T)}$ as an approximation of $\p_u$.  Similar approaches were
also used in~\cite{Berkhin2006, Andersena2006}.  The key difference between our
algorithm and previous approaches is that our \vcd algorithm can
utilize the precomputed approximations to accelerate the online query of PPR
vectors.  As shown in Section~\ref{sec:runtime}, with precomputed
approximations, the query response time can be reduced by $86\%$.

\section{Implementation} \label{sec:implementation}

We now present our complete solution to PPR computation, including its
implementation details.  Our proposed framework is called \emph{\sysname},
which consists of two phases: \emph{offline preprocessing} and \emph{online
query}.  In the offline preprocessing, our MCFP algorithm is used to
approximate the fully PPR (all the vectors $\p_u$ for all $u \in V$).  The size
and computation time of an approximate PPR vector depend on the number of
random walks $R$ in the MCFP algorithm, which in turn depends on the available
memory budget.  In practice, $R$ should be chosen such that the approximate
fully PPR can be computed in reasonable time and can be cached in the allocated
memory.  We call these approximate PPR vectors the ``\emph{PPR index}''.  Recall
that our MCFP algorithm simulates $R$ random walks for each vertex in $V$.
Therefore, the major challenge in generating the PPR index is how to efficiently
simulate all $N \cdot R$ random walks.  We will discuss our solution
implemented on top of DrunkardMob~\cite{Kyrola2013} in
Section~\ref{sec:impl_pre}.

For the PPR query to a vertex $u$ (i.e., to compute $\p_u$), we have two
options: (1) we can return the approximate PPR vector $\ph_u$ in the PPR index
directly; (2) we can use our proposed \vcd algorithm to compute a
new approximate PPR vector $\ph_u$ from the PPR index.  In the first case, to
achieve high accuracy, we need a very large $R$ in the offline preprocessing.
As we will show in Section~\ref{sec:runtime}, when $R$ is large, the
preprocessing is very time consuming.  Also the resulting PPR index is too
large to fit in the main memory which slows down the online query
significantly.  The key idea of \sysname is that, to achieve the same level of
accuracy as the first case, the second case allows us to
use a smaller $R$ in the preprocessing and migrate the computation cost to the online
query.  More importantly, the smaller $R$ allows us to cache the PPR index in
the main memory which accelerates the online query greatly.  We report in
Section~\ref{sec:runtime} that the second
case is much more efficient in handling large batches of queries which is
common in nowadays search engines~\cite{Attenberg2011}.  We will discuss how to execute a large number of
PPR queries efficiently on PowerGraph~\cite{Gonzalez2012} in
Section~\ref{sec:impl_query}.

\subsection{Preprocessing} \label{sec:impl_pre}

In preprocessing, we use our MCFP algorithm to approximate $\p_u$ for every
vertex $u \in V$.  These approximate PPR vectors will constitute the \emph{PPR
index} which accelerates online query.  The MCFP algorithm simulates $R$ random
walks per vertex, thus $N \cdot R$ random walks in total.  Depending on the
memory budget, there are two cases to consider.  If the graph can be cached in
the main memory of a single machine, we can use a simple loop to simulate one
random walk for each source vertex,
where the walk will be confined to one machine and no network or disk latency
will be incurred.  Repeating this procedure $N \cdot R$ times gives us $R$
random walks per vertex.


If the graph is too large to be cached in the main memory of a single machine,
it has to be partitioned and distributed across multiple machines or to be
accessed from disk.  Then at each step, the random walk may incur a network
transmission or disk access.  As network or disk I/O is much slower than memory
access, the random walk simulation will be inefficient.  To address this
challenge, Kyrola proposes an algorithm called \emph{DrunkardMob} to simulate
billions of random walks on massive graphs on a single PC~\cite{Kyrola2013}.
Since it has been shown to be an efficient solution for computing random walks
and is also used in production by Twitter~\cite{Kyrola2013}, we implement our
offline preprocessing based on DrunkardMob, with several important
improvements, as discussed below.

DrunkardMob is built on top of disk-based graph computation systems like
GraphChi~\cite{Kyrola2012}.  To improve the performance, DrunkardMob simulates
a large number of random walks simultaneously and assumes that the states of
all walks can be held in memory.  The states of walks can be seen as a mapping
from vertices to walks: for each vertex, DrunkardMob knows the walks whose last
hop is in that vertex.  In a graph computation system, the edges are usually
sorted by their source vertex, which means we can load a chunk of vertices and
their out-neighbors using a sequential disk I/O.  At each iteration, the graph
computation system performs a sequential scan over the whole graph.  For each
incoming vertex and its out-neighbors, DrunkardMob first retrieves all walks on
that vertex from memory and then moves these walks to the next hop.

To further accelerate the preprocessing, we reimplement DrunkardMob
on our system, VENUS~\cite{Liu2015b} (a disk-based graph computation system in
C++) and extend it in a distributed environment using MPI~\cite{MPI94}.
\sysname distributes the simulation of random walks to a set of workers.  One
of the workers acts as the master to coordinate the remaining slave workers.
The master divides the set of vertices of an input graph into disjoint
intervals.  When there is an idle slave, the master assigns an interval of
vertices to that slave.  Each slave works independently by using DrunkardMob to
simulate $R$ random walks for each vertex in the assigned interval and computes
an approximate PPR vector for each vertex.  Finally, all approximate PPR
vectors compose the PPR index which is used to accelerate the online query.

\paragraph{Remark.}
Simulating random walks on graphs is a staple of many other ranking and
recommendation algorithms~\cite{Fujiwara2012a, Kusumoto2014}.  Our
implementation of random walk simulation is highly competitive compared to
previous solutions on general distribute dataflow systems.  For example, to
simulate $100$ random walks from each
vertex on the Twitter graph, DrunkardMob on VENUS takes 1728.2~seconds while
an implementation on Spark takes 2967~seconds on the same
cluster~\cite{Li2015}.  The detailed setup of our experiments is described in
Section~\ref{sec:eval}.

Next, we give a brief introduction to PowerGraph before we proceed to present
our online query implementation.

\subsection{PowerGraph}

To achieve low query response time, it is essential to utilize the main memory.
Hence we implement the query phase of \sysname on
PowerGraph~\cite{Gonzalez2012}, which is a popular distributed in-memory graph
engine.  Note that it is also possible to implement \sysname on shared-memory
graph processing systems like Galois~\cite{Pingali2011} and
Ligra~\cite{Shun2013}.  However, we do not consider this case, since their
maximum supported graph size is limited by the memory size of a single machine.

Many recent graph engines adopt a flexible \emph{vertex-centric programming
model}~\cite{Malewicz2010, Gonzalez2014, Gonzalez2012, Cheng2015}, which is
quite expressive in encompassing most graph algorithms.  A graph algorithm
can be formulated into a \emph{vertex-program} which can be executed on each vertex
in a parallel fashion, and a vertex can communicate with neighboring vertices
either synchronously or asynchronously.  Each vertex runs its own instance
of the vertex-program which maintains only local information of the graph and
performs user-defined tasks.  The vertex-program can affect the data of other
vertices by sending messages to these vertices.


Compared with other graph engines~\cite{Malewicz2010, Kyrola2012},
PowerGraph~\cite{Gonzalez2012} is empowered by a more sophisticated
vertex-centric \emph{GAS (Gather-Apply-Scatter)} programming model.  In the GAS
model, a vertex-program is split into three conceptual phases: \emph{gather},
\emph{apply}, and \emph{scatter}. In executing the vertex-program for a vertex,
the \emph{gather} phase assembles information from adjacent edges and vertices.
The result is then used in the \emph{apply} phase to update the vertex data.
Finally, the \emph{scatter} phase distributes the new vertex data to the
adjacent vertices.

Next, we show how to formulate the query phase of our proposed framework,
\sysname, on top of PowerGraph.

\subsection{Batch Query} \label{sec:impl_query}

Most existing methods for fully PPR first compute and store the fully PPR in a
database, and at query time, load the PPR vectors from the database
directly~\cite{Bahmani2011, Kyrola2013}.  This is not scalable as the
computation and storage of fully PPR can be prohibitively costly, especially
for large graphs.  For example, our evaluation shows that it takes more than 11
hours to compute the fully PPR for the uk-union graph with 133.6 million
vertices and 5.5 billion edges when $R=2000$, and it needs 1.1~TB to store the
PPR index.  The required time and space increase with the size of the graph and
the number of random walks starting from each vertex.

In practice, it is much desired to allow to allocate the computation and
storage between the offline and online stages, according to the available
budget in memory, response time, and precision.  Motivated by this, our
framework \sysname proposed in this paper computes an approximate of fully PPR
offline, whose computation and storage cost can be decided based on the allowed
time and memory. These precomputed approximate PPR vectors are called \emph{PPR
index}.  At query time, we apply our \vcd algorithm
to efficiently compute a more precise approximation
of the PPR vectors of interest based on the PPR index.  Particularly, our
\vcd algorithm can execute a large number of queries efficiently in
parallel, which is important to modern online services where multiple requests
can come at the same time.  Our key observation is that the computation of
multiple PPR vectors shares the access to the graph and the PPR index, so the
small network packets used to access the graph and the PPR index can be
multiplexed and aggregated into packets with large payload.  This reduces the
average computation time of each PPR vector, because bulk network transfer
workloads are more efficient.  For online query, \sysname buffers the incoming
PPR queries and computes a batch of PPR queries at a time.  As we will show in
Section~\ref{sec:runtime}, our evaluation indicates that we can compute
thousands of PPR queries in several seconds.

We now elaborate further how we can perform batch PPR queries.  We use a vertex
set $S$ to represent a batch of PPR queries: for each $u \in S$, we would like
to approximate $\p_u$.  We can utilize our vertex-program on PowerGraph to
compute $\f_u$ and $\s_u$ for each $u \in S$ as shown in
Algorithm~\ref{alg:query}.  Each instance of the vertex-program maintains two
maps: \fmap and \smap.  For the instance of vertex $u$, $\fmap[v]$ and
$\smap[v]$ represent $\f_v(u)$ and $\s_v(u)$ respectively.  At the beginning of
each iteration, \fmap is initialized with \prevf sent from the previous
iteration (Line~\ref{line:init-prevf}), except that at the first iteration we
let $\f_u = \e_u$ for $u \in S$ (Line~\ref{line:init-query}).  In the method
\apply{}, for each vertex $w$ that $\f_w(u)$ is not zero, we update $\s_w(u)$
and $\f_w(u)$ accordingly.  The method \scatter{}  will be invoked on each
out-edge of vertex $u$, and \fmap will be sent to each out-neighbors of vertex
$u$.  At the end of each iteration, the PowerGraph engine will aggregate \fmap
by their target vertices.

In summary, the online query can be split into two steps.  First, we run the
above vertex-program iteratively to compute $\f_u$ and $\s_u$ for $u \in S$.
Suppose $\ph_u$ is the approximate PPR vector of $u$ stored in the index.  In
the second step, we compute a refined approximation for each $\p_u$ as discussed in
Section~\ref{sec:vc-decomp}:
\[ \pt_u = \s_u + \sum_{v \in V} \f_u(v) \ph_v. \]

\begin{algorithm}[ht]
    \caption{\vcd on PowerGraph} \label{alg:query}
    \class{\vp{$u$}}{
        \KwData{\fmap is a map and $\fmap[v]$ represents $\f_v(u)$}
        \KwData{\smap is a map and $\smap[v]$ represents $\s_v(u)$}
        \proc{\init{$u$, \prevf}}{
            \If{first iteration}{
                \If{$u \in S$}{
                    $\fmap[u] \leftarrow 1$\;
                    \label{line:init-query}
                }
            }\Else{
                $\fmap \leftarrow \prevf$\;
                \label{line:init-prevf}
            }
        }
        \proc{\apply{$u$}}{
            \ForEach{key-value pair $\langle w, t \rangle \in \fmap$}{
                $\smap[w] \leftarrow \smap[w] + c \cdot t$\;
                $\fmap[w] \leftarrow \frac{1-c}{|O(u)|} t$\;
            }
        }
        \proc{\scatter{$u$, edge $(u, v)$}}{
            send \fmap to vertex $v$\;
        }
    }
\end{algorithm}

\subsection{Analysis} \label{sec:analysis}

In this subsection, we analyze the complexity of our MCFP algorithm and
\vcd algorithm.  We also elaborate on some choices we made in designing
\sysname.

Let us start with the MCFP algorithm for a single vertex $u$ described in
Algorithm~\ref{alg:mc-fp}.  If the graph fits in the memory of a single
machine, the time complexity is $O(R/c)$, since we simulate $R$ walks and the
average length of these walks is $1/c$.  More importantly, the size of the
obtained approximation $\ph_u$ is also $O(R/c)$.  So the size of the PPR index
obtained in the preprocessing phase is bounded by $N \cdot R/c$.  On the other
hand, consider the \vcd algorithm for a single vertex $u$ described
in Algorithm~\ref{alg:proj}.  When the number of iterations $T$ is large
enough, the worst case time complexity is $O(N+M)$, since the algorithm works
in the same way as the breadth-first search.  Similarly, in the worst case, the size of $\s_u$,
$\f_u$, and the final approximation $\pt_u$ is $O(N)$.

To compute the PPR vectors for all vertices in a vertex set $S$ at query time,
it is possible to use one of the two Monte-Carlo methods described in
Section~\ref{sec:mc} instead of the \vcd algorithm by starting a number of
random walks from each vertex in $S$.  To make sure that most random walks have
terminated, we need to run a Monte-Carlo method for at least $10$ iterations.
On the other hand, if we use the \vcd
algorithm for online query, we only need two iterations to achieve reasonable
accuracy as shown in our evaluation (Section~\ref{sec:accuracy}).  On
distributed settings, the running time is mainly decided by the bulk network
transfers in each iteration which makes the \vcd algorithm a
preferred approach.  We compare the running time of the MCFP algorithm and
the \vcd algorithm in terms of online query in
Section~\ref{sec:runtime}, and show that the \vcd algorithm is more
efficient.  For example, to execute 10,000 PPR queries on a Twitter social
network with 41 million vertices, the MCFP algorithm takes 179.38~seconds
while the \vcd algorithm takes only 3.64~seconds, which has orders
of magnitude improvement in computational efficiency.

\section{Evaluation} \label{sec:eval}

In this section, we experimentally evaluate our framework, \sysname.  In
Section~\ref{sec:accuracy}, we evaluate the accuracy of our MCFP algorithm and
our \vcd algorithm with various parameters.  In
Section~\ref{sec:runtime}, we evaluate the time and space costs of the offline
preprocessing and also the running time of the online batch query in \sysname.
Let us first introduce the setup of our experiments.

\subsection{Experimental Setup} \label{sec:setup}

In the following experiments, unless
we state otherwise, we set the teleport probability to $0.15$.

\paragraph{Datasets.}
To evaluate \sysname, we run our experiments on six real-world graph datasets as
shown in Table~\ref{tbl:datasets}.
Our datasets are of various types: wiki-Vote is a Wikipedia who-votes-on-whom
network; twitter-2010 is an online social network; and web-BerkStan,
web-Google, uk-1m, and uk-union are web graphs.

\begin{table}[t]
    \small
    \centering
    \caption{Graph datasets} \label{tbl:datasets}
    \begin{tabular}{lrr}
        \toprule
        \textbf{Dataset}                                                & $N$          & $M$           \\
        \midrule
        wiki-Vote~\cite{Leskovec2010}                                   & 7,115        & 103,689       \\
        web-BerkStan~\cite{Leskovec2008a}                               & 685,230      & 7,600,595     \\
        web-Google~\cite{Leskovec2008a}                                 & 875,713      & 5,105,039     \\
        uk-1m~\cite{Boldi2008, Boldi2004}                               & 1,000,000    & 41,247,159    \\
        \addlinespace
        twitter-2010~\cite{Kwak2010}                                    & 41,652,230   & 1,468,365,182 \\
        uk-union~\cite{Boldi2008, Boldi2004}                            & 133,633,040  & 5,507,679,822 \\
        \bottomrule
    \end{tabular}
    \vspace{-0.2in}
\end{table}

\paragraph{Cluster.}
We perform all experiments on a cluster of eight machines, each with two
eight-core Intel Xeon E5-2650 2.60~GHz processors, 377~GB RAM, and 20~TB hard
disk, running Ubuntu~14.04.  All machines are connected via Gigabit Ethernet.

\subsection{Accuracy of Our Algorithms} \label{sec:accuracy}

In this subsection, we evaluate the accuracy of our MCFP algorithm and
\vcd algorithm.

\paragraph{Measurement.}
In most applications of PPR, the common approach is to return the top $k$
ranked vertices of the PPR vector for the vertex in query~\cite{Kyrola2013}.
For example, in Twitter's ``Who to Follow'' recommendation service, only the
top-ranked users will be recommended~\cite{Gupta2013}.  So to compare the exact
and approximate PPR vectors personalized to vertex $u$, it is important to
measure the difference between the top $k$ ranked vertices of exact PPR vector
$p_u$ and approximate PPR vector $\ph_u$.  Let $T_k^u$ be the set of vertices
having the $k$ highest scores in the PPR vector $\p_u$.  We approximate $T_k^u$
by $\widehat{T_k^u}$, which is the set of vertices having the $k$ highest
approximated scores in $\ph_u$ obtained from our algorithms or competitors.
The \emph{relative aggregated goodness (RAG)}~\cite{Singitham2004} measures how
well the approximate top-$k$ set performs in finding a set of vertices with
high PPR scores.  RAG calculates the sum of exact PPR values in the approximate
set compared to the maximum value achievable (by using the exact top-k set
$T_k^u$):
\[
    \RAG(k,u) = \frac{\sum_{v \in \widehat{T_k^u}} \p_u(v)}{\sum_{v \in T_k^u}
    \p_u(v)}.
\]
Note that the higher the RAG, the higher is the accuracy.  Also, RAG and its
variant have been widely used in previous work~\cite{Fogaras2005, Sarlos2006,
Bahmani2011}.  Since the degrees in social or web graphs follow power law
distributions, the vertices in tested graphs have widely different degrees.  In
order to observe the behavior of the algorithms for different degrees, we
divide the vertices in each tested graph into $B$ buckets, with bucket $i$
including all vertices whose out-degrees are in the interval $[2^{i-1}, 2^i)$
for $i=1, 2, \ldots, B-1$ and bucket $B$ including all vertices whose
out-degrees are in the interval $[2^{B-1}, \infty)$.  Since when $B=10$, bucket
$B$ only contains a very small number of vertices, so we set $B$ to $10$ in our
evaluation.  For each tested graph, we choose $10$ vertices from each bucket
randomly.  Then we use the power iteration method to compute the ground truth
PPR vectors for all $100$ selected vertices, denoted by $\p_u$ for a vertex
$u$.  Then, we compute the average RAG for these vertices at different values
of $k$.  Because the computation of the ground truth is very costly, we
evaluate the accuracy on four small graphs: wiki-Vote, web-BerkStan,
web-Google, and uk-1m.

\paragraph{Comparison of Monte-Carlo methods.}
In Figure~\ref{fig:mc-rag-200}, we evaluate the
accuracy of our proposed MCFP algorithm and the state-of-the-art method
described in Algorithm~\ref{alg:mc-ep} (and we denote it by MCEP).  From
Figure~\ref{fig:mc-rag-200}, to achieve RAG larger than $0.99$, it suffices to set
$R=2000$ for MCFP, when $k$ is $200$.  Also, we observe that the existing MCEP
algorithm is less accurate given the same value of $R$, the number of random
walks starting from each vertex.

Note that given $R$ random walks, our MCFP algorithm generates $R/c \approx
6.7R$ dependent sample points, since the average length of each walk is $1/c
\approx 6.7$.  In Figure~\ref{fig:k-rag}, we observe that our MCFP algorithm
with only $1,000$ random walks achieves the same level of accuracy as the
existing MCEP algorithm with $6,700$ random walks.  This means that the
dependent sample points from the MCFP algorithm and
the independent sample points from the MCEP algorithm actually have almost the
same effect in approximating PPR vector.  From Figure~\ref{fig:k-rag}, we also
observe that the RAG decreases as $k$ increases.  The reason is that when we
have a larger value of $k$, it is more difficult to approximate the top-$k$
set, because lower ranked vertices has smaller difference in PPR scores.

\begin{figure}[t]
    \centering
    \ref{legend:mc-rag-200}
    \\
    \subfigure[wiki-Vote]{
        \begin{tikzpicture}
            \begin{semilogxaxis}[
                    width=\textwidth*0.25,
                    grid=both,
                    ymin=0,
                    legend columns=-1,
                    legend to name=legend:mc-rag-200,
                    xtick={1, 10, 100, 1000, 10000},
                    minor y tick num=4,
                    xlabel=$R$, ylabel=RAG
                ]
                \addplot table[x=R,y=wiki-Vote-1]{./tables/mc-k200.table};
                \addplot table[x=R,y=wiki-Vote-2]{./tables/mc-k200.table};
                \legend{MCEP, MCFP}
            \end{semilogxaxis}
        \end{tikzpicture}
    }
    \subfigure[web-BerkStan]{
        \begin{tikzpicture}
            \begin{semilogxaxis}[
                    width=\textwidth*0.25,
                    grid=both,
                    ymin=0,
                    xtick={1, 10, 100, 1000, 10000},
                    minor y tick num=4,
                    xlabel=$R$, ylabel=RAG
                ]
                \addplot table[x=R,y=web-BerkStan-1]{./tables/mc-k200.table};
                \addplot table[x=R,y=web-BerkStan-2]{./tables/mc-k200.table};
            \end{semilogxaxis}
        \end{tikzpicture}
    }
    \subfigure[web-Google]{
        \begin{tikzpicture}
            \begin{semilogxaxis}[
                    width=\textwidth*0.25,
                    grid=both,
                    ymin=0,
                    xtick={1, 10, 100, 1000, 10000},
                    minor y tick num=4,
                    xlabel=$R$, ylabel=RAG
                ]
                \addplot table[x=R,y=web-Google-1]{./tables/mc-k200.table};
                \addplot table[x=R,y=web-Google-2]{./tables/mc-k200.table};
            \end{semilogxaxis}
        \end{tikzpicture}
    }
    \subfigure[uk-1m]{
        \begin{tikzpicture}
            \begin{semilogxaxis}[
                    width=\textwidth*0.25,
                    grid=both,
                    ymin=0,
                    xtick={1, 10, 100, 1000, 10000},
                    minor y tick num=4,
                    xlabel=$R$, ylabel=RAG
                ]
                \addplot table[x=R,y=uk-2007-05@1000000-1]{./tables/mc-k200.table};
                \addplot table[x=R,y=uk-2007-05@1000000-2]{./tables/mc-k200.table};
            \end{semilogxaxis}
        \end{tikzpicture}
    }
    \caption{Effect of $R$ on the MCFP algorithm ($k=200$)}
    \label{fig:mc-rag-200}
    \vspace{-0.2in}
\end{figure}
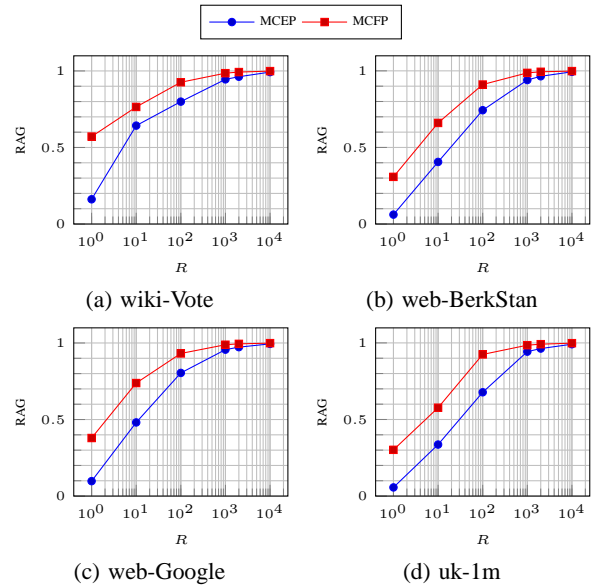

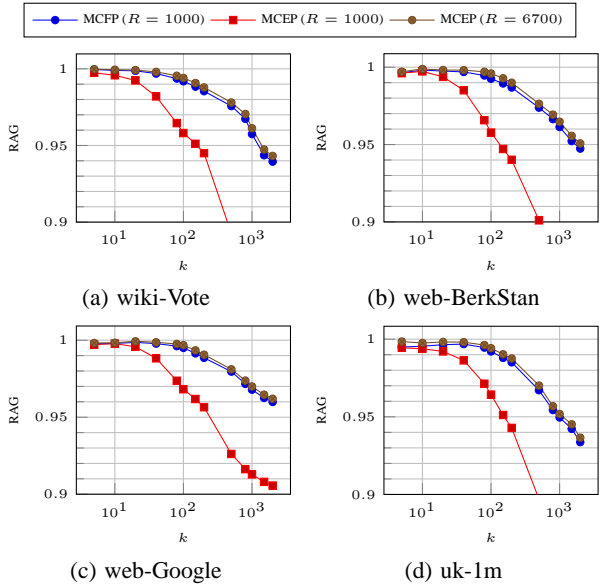
\begin{figure}[t]
    \centering
    \ref{legend:k-rag}
    \\
    \subfigure[wiki-Vote]{
        \begin{tikzpicture}
            \begin{semilogxaxis}[
                    width=\textwidth*0.25,
                    grid=both,
                    ymin=0.9,
                    legend columns=-1,
                    legend to name=legend:k-rag,
                    xtick={1,10,100,1000,5000},
                    minor y tick num=4,
                    xlabel=$k$, ylabel=RAG,
                ]
                \addplot table[x=k,y=fp-r1000]{./tables/wiki-Vote-k-rag.table};
                \addplot table[x=k,y=ep-r1000]{./tables/wiki-Vote-k-rag.table};
                \addplot table[x=k,y=ep-r6700]{./tables/wiki-Vote-k-rag.table};
                \legend{MCFP ($R=1000$), MCEP ($R=1000$), MCEP ($R=6700$)}
            \end{semilogxaxis}
        \end{tikzpicture}
    }
    \subfigure[web-BerkStan]{
        \begin{tikzpicture}
            \begin{semilogxaxis}[
                    width=\textwidth*0.25,
                    grid=both,
                    ymin=0.9,
                    xtick={1,10,100,1000,5000},
                    minor y tick num=4,
                    xlabel=$k$, ylabel=RAG,
                ]
                \addplot table[x=k,y=fp-r1000]{./tables/web-BerkStan-k-rag.table};
                \addplot table[x=k,y=ep-r1000]{./tables/web-BerkStan-k-rag.table};
                \addplot table[x=k,y=ep-r6700]{./tables/web-BerkStan-k-rag.table};
            \end{semilogxaxis}
        \end{tikzpicture}
    }
    \subfigure[web-Google]{
        \begin{tikzpicture}
            \begin{semilogxaxis}[
                    width=\textwidth*0.25,
                    grid=both,
                    ymin=0.9,
                    xtick={1,10,100,1000,5000},
                    minor y tick num=4,
                    xlabel=$k$, ylabel=RAG,
                ]
                \addplot table[x=k,y=fp-r1000]{./tables/web-Google-k-rag.table};
                \addplot table[x=k,y=ep-r1000]{./tables/web-Google-k-rag.table};
                \addplot table[x=k,y=ep-r6700]{./tables/web-Google-k-rag.table};
            \end{semilogxaxis}
        \end{tikzpicture}
    }
    \subfigure[uk-1m]{
        \begin{tikzpicture}
            \begin{semilogxaxis}[
                    width=\textwidth*0.25,
                    grid=both,
                    ymin=0.9,
                    xtick={1,10,100,1000,5000},
                    minor y tick num=4,
                    xlabel=$k$, ylabel=RAG,
                ]
                \addplot table[x=k,y=fp-r1000]{./tables/uk-2007-05-1000000-k-rag.table};
                \addplot table[x=k,y=ep-r1000]{./tables/uk-2007-05-1000000-k-rag.table};
                \addplot table[x=k,y=ep-r6700]{./tables/uk-2007-05-1000000-k-rag.table};
            \end{semilogxaxis}
        \end{tikzpicture}
    }
    \caption{Comparison of the MCFP algorithm and MCEP algorithm with varying
    top set size $k$}
    \label{fig:k-rag}
    \vspace{-0.2in}
\end{figure}

\paragraph{Effectiveness of the \vcd algorithm.}
As described in Section~\ref{sec:alg}, our \vcd algorithm computes
a new PPR vector based on the PPR index whose precision is decided by
the number of random walks $R$ used in the preprocessing.  Obviously, if the
precision of the PPR index is low, we have to run the \vcd
algorithm for more iterations.  If the precision of the PPR index is
high enough for online query, then \sysname can return the PPR vector in the
index directly.  In this experiment, we first use the preprocessing
procedure described in Section~\ref{sec:impl_pre} to generate the PPR index
with $R=10$ or $100$.  Then, we use the batch query procedure based on
our \vcd algorithm described in Section~\ref{sec:impl_query} to
execute the PPR queries.  The results are shown in
Figure~\ref{fig:proj-rag-200}.  Here, when $T=0$, it means that the batch query
procedure returns the PPR vector stored in the PPR index directly.  Also,
when $R=0$, it means we use the \vcd algorithm directly without the
precomputed PPR index.  On all four graphs, to achieve RAG larger than
$0.99$, we need $T=7, 5, 2$ iterations when $R=0, 10, 100$ respectively.  The
results confirm with our hypothesis: with more random walks used in the
preprocessing phase, we need less iterations during the batch query phase to
achieve the same level of accuracy.  However, more random walks also increase
the overheads of the preprocessing phase.  Let us examine the time and space
costs of the preprocessing phase and also the running time of the batch query
phase in the next subsection.

\begin{figure}[t]
    \centering
    \ref{legend:proj-rag-200}
    \\
    \subfigure[wiki-Vote]{
        \begin{tikzpicture}
            \begin{axis}[
                    width=\textwidth*0.25,
                    grid=both,
                    legend style={anchor=south, at={(0.5,1.03)}},
                    legend columns=-1,
                    legend to name=legend:proj-rag-200,
                    ymin=0.5,
                    xtick=data,
                    xlabel=$T$, ylabel=RAG
                ]
                \addplot table[x=T,y=wiki-Vote-R0]{./tables/proj-k200.table};
                \addplot table[x=T,y=wiki-Vote-R10]{./tables/proj-k200.table};
                \addplot table[x=T,y=wiki-Vote-R100]{./tables/proj-k200.table};
                \legend{$R=0$, $R=10$, $R=100$}
            \end{axis}
        \end{tikzpicture}
    }
    \subfigure[web-BerkStan]{
        \begin{tikzpicture}
            \begin{axis}[
                    width=\textwidth*0.25,
                    grid=both,
                    ymin=0.5,
                    xtick=data,
                    xlabel=$T$, ylabel=RAG
                ]
                \addplot table[x=T,y=web-BerkStan-R0]{./tables/proj-k200.table};
                \addplot table[x=T,y=web-BerkStan-R10]{./tables/proj-k200.table};
                \addplot table[x=T,y=web-BerkStan-R100]{./tables/proj-k200.table};
            \end{axis}
        \end{tikzpicture}
    }
    \subfigure[web-Google]{
        \begin{tikzpicture}
            \begin{axis}[
                    width=\textwidth*0.25,
                    grid=both,
                    ymin=0.5,
                    xtick=data,
                    xlabel=$T$, ylabel=RAG
                ]
                \addplot table[x=T,y=web-Google-R0]{./tables/proj-k200.table};
                \addplot table[x=T,y=web-Google-R10]{./tables/proj-k200.table};
                \addplot table[x=T,y=web-Google-R100]{./tables/proj-k200.table};
            \end{axis}
        \end{tikzpicture}
    }
    \subfigure[uk-1m]{
        \begin{tikzpicture}
            \begin{axis}[
                    width=\textwidth*0.25,
                    grid=both,
                    ymin=0.5,
                    xtick=data,
                    xlabel=$T$, ylabel=RAG
                ]
                \addplot table[x=T,y=uk-2007-05@1000000-R0]{./tables/proj-k200.table};
                \addplot table[x=T,y=uk-2007-05@1000000-R10]{./tables/proj-k200.table};
                \addplot table[x=T,y=uk-2007-05@1000000-R100]{./tables/proj-k200.table};
            \end{axis}
        \end{tikzpicture}
    }
    \caption{Effect of $T$ on the \vcd algorithm ($k=200$)}
    \label{fig:proj-rag-200}
    \vspace{-0.2in}
\end{figure}
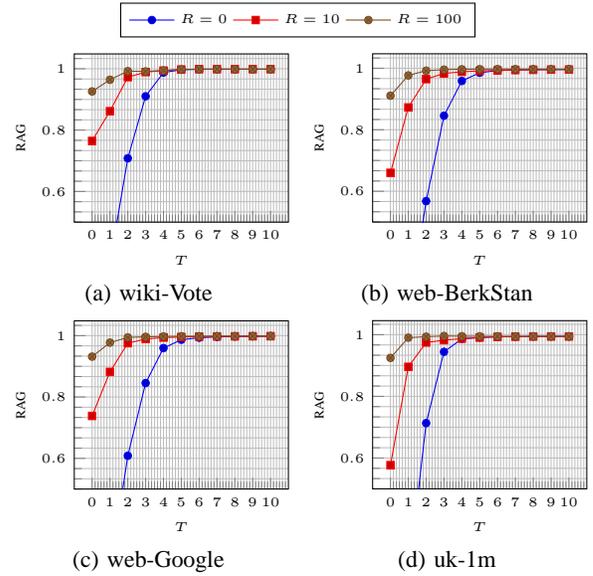

\subsection{Performance} \label{sec:runtime}

In this subsection, we first evaluate the space and time costs of the
preprocessing phase of \sysname on two large graphs twitter-2010 and uk-union.
Then, we compare the online query of \sysname with other previous
state-of-the-art methods.

\paragraph{Offline Preprocessing.}
We show the preprocessing costs in Table~\ref{tbl:preprocessing}.  The
preprocessing time is {\em sublinear} to the number of random walks $R$
starting from each vertex, which is an attractive property.  When we set
$R=100$, the preprocessing time of the uk-union graph is 53~mins and the index
size is 1.6X compared to the graph size.  As we will show later, we can achieve
very fast online query when we set $R=100$.

Fogaras et al.~\cite{Fogaras2005} propose to use the Monte-Carlo methods to
first compute the fully PPR, and then return the
PPR vectors directly as query results.  This approach can be considered as a
special case of our \sysname when $R$ is large.  In the preprocessing, when $R$ is
large enough, the precomputed approximate PPR vectors are already accurate
enough and can be returned directly as online query results.  If we want to
achieve RAG larger than $0.99$, this special case can be achieved by setting
$R$ to $2000$.  However, as shown in Table~\ref{tbl:preprocessing}, in this
case, the preprocessing takes $11.6$~hours for the uk-union graph to finish,
which makes it undesirable for large graphs.  More importantly,
since the index size increases almost linearly with $R$, when $R=2000$, the PPR
index ($1.48$~TB on the disk for the twitter-2010 graph) cannot fit in the
distributed memory (around $2.9$~TB in our cluster) due to the additional
storage overhead of the hash tables.
This significantly slows down the online
query as we will show later.  This means that for very large graphs, computing
and storing the fully PPR with large $R$ is impractical.  In contrast, \sysname
only computes and caches a light-weight PPR index which allows fast online
query.

\begin{table*}[t]
    \centering
    \small
    \caption{Preprocessing Costs} \label{tbl:preprocessing}
    \begin{tabular}{lrlrrr}
        \toprule
        \textbf{Dataset}   & \textbf{Data Size} & \textbf{Type}      & $R=10$   & $R=100$  & $R=2000$   \\
        \midrule
        twitter-2010       & 25 GB              & Preprocessing Time & 933.8 s  & 1728.2 s & 4.5 hours  \\
                           &                    & Index Size on Disk & 12.4 GB  & 95.5 GB  & 1.48 TB    \\
        \addlinespace
        uk-union           & 93 GB              & Preprocessing Time & 2087.4 s & 3187.2 s & 11.6 hours \\
                           &                    & Index Size on Disk & 29.7 GB  & 148.0 GB & 1.1 TB     \\
        \bottomrule
    \end{tabular}
    \vspace{-0.2in}
\end{table*}

\paragraph{Online Query.}
In Section~\ref{sec:accuracy}, we observe that to achieve RAG larger than
$0.99$, when $R$ is set to $0$, $10$, and $100$ in the
preprocessing, the \vcd algorithm should set the number of
iterations $T$ to $7$, $5$, and $2$ respectively.  We compare the query
response time by using different values of $R$ in the preprocessing, and
these results are shown in Figure~\ref{fig:query}.  We can observe that when $R$
increases, the query response time decreases because it needs less iterations to
achieve the same accuracy.
Note that $R=0$ means that \sysname runs without the PPR index.  In comparison,
with the PPR index generated by setting $R=100$, the query response time is
reduced by $86\%$ on the uk-union graph when the number of queries is $10,000$
vertices.
Moreover, when the number of queries increases, the query response time
increases very slowly.  This means our batch query procedure is very efficient by
sharing the computation and access to the underlying graph and index.


Finally, we compare our \sysname with other algorithms for online query:
\begin{itemize}
    \item PI: We implement the naive power iteration method~\cite{Page1999} on
        PowerGraph.
    \item MCFP: We also implement the MCFP algorithm as described in
        Section~\ref{sec:mc} for online query on PowerGraph.  To achieve RAG
        $ > 0.99$, we set $R$ to $2000$.
    \item Fully PPR (FPPR): As discussed above, our preprocessing phase can
        compute all PPR vectors offline.  Then at the online query, the PPR
        vectors in the PPR index can be returned directly.  However, in this
        case, the PPR index cannot fit in the main memory, so we store the PPR
        index in a constant key/value storage library called {\tt
        sparkey}\footnote{\url{https://github.com/spotify/sparkey}}.
\end{itemize}
The results are summarized in Table~\ref{tbl:perf}.  \sysname outperforms PI
and MCFP significantly, especially when the number of queries is large.  For
example, as shown in Table~\ref{tbl:perf}, the MCFP algorithm takes 179.38
seconds to compute 10,000 PPR vectors on twitter-2010, while our
\vcd algorithm takes just 9.06 and 3.64 seconds when $R=0$ and
$100$ in preprocessing respectively.  Note that the power iteration method
cannot even handle multiple queries due to its large space requirement.  When
the number of queries is one, FPPR has the fastest response time.  The reason is that all
PPR vectors are already precomputed in the preprocessing phase and it only
requires one disk I/O to retrieve the query result from {\tt sparkey}.
However, as the number of queries increases, the query response time of FPPR
increases linearly, because the throughput of {\tt sparkey} is
bounded by disk I/O when the PPR index is larger than the main memory.

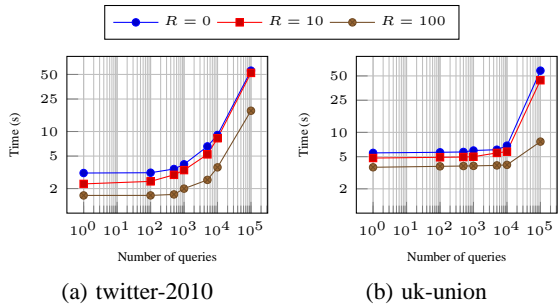
\begin{figure}[t]
    \centering
    \ref{legend:query}
    \\
    \subfigure[twitter-2010]{
        \begin{tikzpicture}
            \begin{loglogaxis}[
                    width=\textwidth*0.24,
                    grid=both,
                    legend style={anchor=south, at={(0.5,1.03)}},
                    legend columns=-1,
                    legend to name=legend:query,
                    xtick={1, 10, 100, 1000, 10000, 100000},
                    log y ticks with fixed point,
                    ytick={2, 5, 10, 25, 50},
                    ymin=1,
                    xlabel=Number of queries, ylabel=Time (s)
                ]
                \addplot table[x=k,y=R0]{./tables/query-twitter.table};
                \addplot table[x=k,y=R10]{./tables/query-twitter.table};
                \addplot table[x=k,y=R100]{./tables/query-twitter.table};
                \legend{$R=0$, $R=10$, $R=100$}
            \end{loglogaxis}
        \end{tikzpicture}
    }
    \subfigure[uk-union]{
        \begin{tikzpicture}
            \begin{loglogaxis}[
                    width=\textwidth*0.24,
                    grid=both,
                    xtick={1, 10, 100, 1000, 10000, 100000},
                    log y ticks with fixed point,
                    ytick={2, 5, 10, 25, 50},
                    ymin=1,
                    xlabel=Number of queries, ylabel=Time (s)
                ]
                \addplot table[x=k,y=R0]{./tables/query-uk.table};
                \addplot table[x=k,y=R10]{./tables/query-uk.table};
                \addplot table[x=k,y=R100]{./tables/query-uk.table};
            \end{loglogaxis}
        \end{tikzpicture}
    }
    \caption{Running time of the online query with varying the number of
    queries}
    \label{fig:query}
    \vspace{-0.2in}
\end{figure}

\begin{table*}[t]
    \centering
    \small
    \caption{Comparison of \sysname with other PPR computation algorithms} \label{tbl:perf}
    \begin{tabular}{llrrrrrr}
        \toprule
        \multirow{2}{*}{\textbf{Dataset}} & \textbf{Number of} &
        \multirow{2}{*}{\textbf{PI}} & \multirow{2}{*}{\textbf{MCFP}} &
        \multirow{2}{*}{\textbf{FPPR}} &
        \multicolumn{3}{c}{\textbf{\sysname}} \\
        \cmidrule(lr){6-8}
                         & \textbf{Queries} &        &           &           & $R=0$   & $R=10$  & $R=100$ \\
        \midrule
        twitter-2010     & 1                & 95.2 s & 7.83 s    & 25 ms     & 3.10 s  & 2.28 s  & 1.65 s  \\
                         & 100              & N/A    & 9.45 s    & 2.50 s    & 3.13 s  & 2.45 s  & 1.65 s  \\
                         & 10,000           & N/A    & 179.38 s  & 117.50 s  & 9.06 s  & 8.74 s  & 3.64 s  \\
                         & 100,000          & N/A    & 1739.74 s & 978.25 s  & 55.65 s & 52.20 s & 17.96 s \\
        \addlinespace
        uk-union         & 1                & 320 s  & 15.47 s   & 18 ms     & 5.57 s  & 4.81 s  & 3.70 s  \\
                         & 100              & N/A    & 16.85 s   & 1.75 s    & 5.67 s  & 4.91 s  & 3.80 s  \\
                         & 10,000           & N/A    & 20.14 s   & 125.25 s  & 6.84 s  & 5.78 s  & 3.96 s  \\
                         & 100,000          & N/A    & 50.20 s   & 1037.25 s & 58.10 s & 44.19 s & 7.68 s  \\
        \bottomrule
    \end{tabular}
    \vspace{-0.2in}
\end{table*}

\section{Related Work} \label{sec:related}

There are extensive studies on Personalized PageRank computation.  Here, we
highlight some examples.  The idea of Personalized PageRank is first
proposed in~\cite{Page1999}.  The simplest method for computing PPR is by the
power iteration method~\cite{Page1999}.  To compute $\p_u$, it starts with
$\p_u = \e_u$ and repeatedly performs the update following
Equation~(\ref{eq:ppr}).  Since the complexity of every iteration to compute
only one PPR vector is $O(N+M)$, the power iteration method is prohibitively
expensive for large graphs.  Many approximation techniques have been proposed
to speed up the PPR computation since then.

The seminal work by Jeh and Widom~\cite{Jeh2003} proposed the Hub Decomposition
algorithm which approximates the PPR vectors for a small hub set $H$ of
high-PageRank vertices.  However, to achieve full personalization, the hub set
needs to include all vertices which requires $O(N^2)$ space, clearly
impractical for large problems.   Our \vcd algorithm is partially
inspired by the decomposition theorem in~\cite{Jeh2003}.  The key idea of our
\vcd algorithm is that it can utilize the precomputed PPR index to
provide fast online query for any vertex.  Compared to the Hub Decomposition
algorithm, the size of the PPR index in \sysname is much more compact and thus
can be cached in distributed memory, which is critical to efficient online
query.

Note that most existing methods for PPR computation are designed for a single
machine, and are therefore limited by its restricted computational
power~\cite{Sarlos2006, Fujiwara2012, Lofgren2014, Maehara2014b, Shin2015}.
Fogaras et al.~\cite{Fogaras2005} proposed the Monte-Carlo End-Point algorithm
which is the first scalable solution that achieves full personalization,
although they do not provide an implementation.  To the best of our knowledge,
there are only two scalable implementations for~\cite{Fogaras2005}.  The first
system is designed for the MapReduce model, which aims to optimize the I/O
efficiency~\cite{Bahmani2011}.  The other system, called
DrunkardMob~\cite{Kyrola2013}, is designed for single-machine disk-based graph
engines like GraphChi~\cite{Kyrola2012} and VENUS~\cite{Liu2015b}.  Although
DrunkardMob works on a single machine, since it stores the graph on disk, it
can still handle graphs larger than the main memory.  The two implementations
can be used to compute the PPR vectors offline.  However, since they both rely
on disk heavily, they cannot handle online query efficiently unless they
precompute all PPR vectors which is very costly.  In this paper, we propose a
more efficient Monte-Carlo Full-Path (MCFP) algorithm which utilizes the full
trajectory of each random walk.  Furthermore, we propose the \vcd
algorithm which can execute PPR queries online based on the PPR index obtained
from the MCFP algorithm.  Our \vcd algorithm could also be realized
on disk-based graph engines like GraphChi~\cite{Kyrola2012} and
VENUS~\cite{Liu2015b} that support vertex-centric programming. Since at each
iteration, a disk-based graph engine will scan the entire graph on disk, the
I/O cost is at least $\frac{M}{B}$, where $B$ is the size (measured in edge) of
block transfer.  So due to the excessive I/O cost, the \vcd
algorithm, if implemented on disk-based graph engines, could be difficult to
handle online PPR queries very efficiently.  Fujiwara et al.~\cite{Fujiwara2012}
proposes a method to compute a single PPR vector via non-iterative approach
based on sparse matrix representation.  Similar to \sysname, their solution
also separates the computation into an offline phase and an online phase.
However, in the offline phase, their approach needs to permute the adjacent
matrix and precompute the QR decomposition of the matrix. For large matrix,
this can be very time consuming.  For example, even for a graph with only 0.35
million vertices and 1.4 million edges, the precomputation takes more than two
hours which makes their method not suitable for large-scale graphs.

Avrachenkov et al.~\cite{Avrachenkov2007} proposed a Monte-Carlo complete path
algorithm for computing the global PageRank which also utilizes the full
trajectory of each random walk.  Bahmani et al.~\cite{Bahmani2010} extended the
Monte-Carlo method in~\cite{Avrachenkov2007} to incremental graphs and
Personalized PageRank.  Their approach stores a small number of precomputed
random walks for each vertex, and then stitch short random walks to answer
online (Personalized) PageRank queries~\cite{Bahmani2010}.  However, this
algorithm is not suitable for distributed graph engines for several reasons:
(1) it requires random accesses to the graph data and precomputed index, (2)
its algorithmic logic is not compatible with the vertex-centric programming
model, and 3) it may require a large number of iterations.

Buehrer and Chellapilla~\cite{Buehrer2008} designed a graph
compression algorithm for web graphs.  To compute (Personalized) PageRank or
other random walk based computations, they proposed to first compress the web
graph, and then perform the computation on the compressed graph using
the power iteration method.  However, this approach has two disadvantages: (1)
it cannot be extended to weighted graphs like our approach; (2) its efficiency
can degrade significantly on social networks since compressing social graphs is
usually more costly and the compression ratios on social networks are much
worse compared to web graphs~\cite{Boldi2011}.  Also, to our knowledge, there
is no available implementation of the algorithm in~\cite{Buehrer2008}.

\section{Conclusion} \label{sec:conclusion}

In this paper, we present our system, \sysname, which adopts a novel framework
for online PPR computation on distributed graph engines.  \sysname uses the
MCFP algorithm to compute a light-weight PPR index, and then uses the \vcd
algorithm to compute PPR vectors by a linear combination of the PPR index in an
online manner.  Our evaluation shows that our MCFP algorithm provides a more
accurate approximation compared to the existing Monte-Carlo End-Point
algorithm~\cite{Fogaras2005} by simulating the same number of random walks.
Extensive experiments on two large-scale real-world graphs show that \sysname
is quite scalable in balancing offline preprocessing and online query, and is
capable of computing tens of thousands of PPR vectors in an order of seconds.
We believe PowerWalk can be readily extended for many large-scale random walk
models~\cite{Li2015, Wu2012}.

{\small
\bibliographystyle{abbrv}
\bibliography{library,ref}
}

\ifdefined\techreport
\begin{appendix}

\section{Proof of Theorem~\ref{thm:fp-bound}} \label{sec:proof-1}

To prove Theorem~\ref{thm:fp-bound}, we first need the following result on
random walks~\cite[Theorem~2.1]{Gillman1998}.

\begin{thm}[The Chernoff bound for Random Walks] \label{thm:gillman}
    Consider the random walk following the transition matrix $\P$ on a graph $G
    = (V,E)$ with initial distribution $\q$. Let $S \subset V$,
    and let $t_n$ be the number of visits to $S$ in $n$ steps. For any $\gamma
    \geq 0$,
    \begin{align*}
        & \hspace{-0.3in} \Pr\left[t_n - n \sum_{v \in S} \pivec(v) \geq
            \gamma\right] \leq \\
        & \quad \left(1+\frac{\gamma\epsilon}{10n}\right) \sqrt{\sum_{w \in V}
            \left(\frac{\q(w)}{\sqrt{\pivec(w)}}\right)^2}
            \exp\left(\frac{-\epsilon \gamma^2}{20n}\right),
    \end{align*}
    where $\pivec$ is the stationary distribution and $\epsilon$ is the
    eigenvalue gap of stochastic matrix $\P$.
\end{thm}

\begin{proof}[Proof of Theorem~\ref{thm:fp-bound}]
    For a given vertex $u$, we define the matrix $\P$ as
    $$ \P = (1-c) \A + c \mathds{1} \e_u, $$
    where $\mathds{1}$ is a column vector whose elements are all $1$'s.  It is
    easy to see that P is row-stochastic and Equation~\eqref{eq:ppr} can be
    rewritten as:
    $$ \p_u = \p_u \P. $$

    As discussed in Section~\ref{sec:mc}, the $R$ random walks in
    Algorithm~\ref{alg:mc-fp} can be seen as one long random walk from vertex
    $u$ following the transition matrix $\P$.  Since the average length of each
    walk is $1/c$, the total length of all $R$ random walks or the length of
    the long random walk is $n \approx R/c$.  By using
    Theorem~\ref{thm:gillman}, we have
    \begin{align}
        & \Pr[\ph_u(v) - \p_u(v) \geq \gamma] = \Pr[\x_n(v) - n \p_u(v) \geq n
            \gamma] \nonumber \\
        & \quad \leq \left(1+\frac{n\gamma\epsilon}{10n}\right) \sqrt{\sum_{w
            \in V} \left(\frac{\q(w)}{\sqrt{\p_u(w)}}\right)^2}
            \exp\left(\frac{-\epsilon (n\gamma)^2}{20n}\right).
        \label{eq:2-1}
    \end{align}

    Since $\P$ is a row-stochastic matrix, the largest eigenvalue $\lambda_1$
    is $1$, while the second largest eigenvalue, denoted as $\lambda_2$, is
    $1-c$~\cite{Haveliwala2003}.  So the eigenvalue gap $\epsilon$ is
    $\lambda_1 - \lambda_2 = c$.  Since the merged long random walk always
    starts from $u$, the initial distribution is $\e_u$.  From
    Equation~\eqref{eq:ppr}, it is easy to see that $\p_u(u) \geq c$.  Hence,
    we have
    \[
        \sqrt{\sum_{w \in V} \left(\frac{\q(w)}{\sqrt{\p_u(w)}}\right)^2} =
        \frac{1}{\sqrt{c}}.
    \]
    Now we can apply the above results into Equation~\eqref{eq:2-1}, and we
    have
    \begin{align*}
        & \hspace{-0.2in} \Pr[\ph_u(v) - \p_u(v) \geq \gamma] \\
        & \quad \leq \left(1+\frac{n \gamma c}{10n}\right) \frac{1}{\sqrt{c}}
            \exp\left(\frac{-c (n\gamma)^2}{20n}\right) \\
        & \quad = \frac{1}{\sqrt{c}} \left(1+\frac{\gamma c}{10}\right)
            \exp\left(\frac{-\gamma^2 R}{20}\right),
    \end{align*}
    which completes the bound of over-estimation.

    Applying Theorem~\ref{thm:gillman} to the set $V \backslash \set{v}$ gives
    the same bound for the probability of under-estimation
    \begin{align*}
        & \Pr[\ph_u(v) - \p_u(v) \leq -\gamma] \\
        & \quad = \Pr[\x_n(v) - n \p_u(v) \leq - n \gamma] \\
        & \quad = \Pr\left[\sum_{w \in V \backslash \set{v}} \x_n(w) - n \sum_{w \in
        V \backslash \set{v}} \p_u(w) \geq n \gamma\right]\\
        & \quad \leq \frac{1}{\sqrt{c}} \left(1+\frac{\gamma c}{10}\right)
            \exp\left(\frac{-\gamma^2 R}{20}\right),
    \end{align*}
    which completes the proof.
\end{proof}

\section{Proof of Theorem~\ref{thm:vd}} \label{sec:proof-2}

\begin{proof}
    To prove Theorem~\ref{thm:vd}, we claim that for $k = 0, 1, \ldots, T$
    \[
        \proj(u, T) = \s_u^{(T-k)} + \sum_{v \in V} \f_u^{(T-k)}(v) \decomp(u,k).
    \]
    The proof is by induction on $k$.  The case for $k=0$ is obvious.  Suppose
    the claim is true for some $k \geq 0$.  We show it holds for $k+1$ as
    follows:
    \begin{align*}
        & \quad \proj(u, T) \\
        & = \s_u^{(T-k)} + \sum_{v \in V} \f_u^{(T-k)}(v) \decomp(v,k) \\
        & = \left(\s_u^{(T-k-1)} + \sum_{w \in V} c \cdot \f_u^{(T-k-1)}(w)\e_w\right) \\
        & + \sum_{v' \in V} \left(\sum_{w \in V} \frac{1-c}{|O(w)|} \sum_{v \in
        O(w)} \f_u^{(T-k-1)}(w) \e_v \right)(v') \decomp(v',k) \\
        & = \left(\s_u^{(T-k-1)} + \sum_{w \in V} c \cdot \f_u^{(T-k-1)}(w)\e_w\right) \\
        & + \left(\sum_{w \in V} \frac{1-c}{|O(w)|} \sum_{v \in
        O(w)} \f_u^{(T-k-1)}(w) \decomp(w,k) \right) \\
        & = \s_u^{(T-k-1)} \\
        & + \sum_{w \in V} \f_u^{(T-k-1)}(w) \left(c \cdot \e_w +
        \frac{1-c}{|O(w)|} \sum_{v \in O(w)} \decomp(w,k) \right) \\
        & = \s_u^{(T-(k+1))} + \sum_{w \in V} \f_u^{(T-(k+1))}(w) \decomp(w,k+1)
    \end{align*}
    When $k=T$, we have
    \begin{align*}
        \proj(u, T) & = \s_u^{(0)} + \sum_{v \in  V} \f_u^{(0)}(v) \decomp(v, T) \\
                    & = \decomp(u, T).
    \end{align*}
\end{proof}

\end{appendix}

\fi

\end{document}